\documentclass[11pt]{article}
\usepackage[a4paper, margin=1in]{geometry}
\usepackage[dvipsnames]{xcolor}
\usepackage{microtype}
\usepackage{amsmath}
\usepackage{amsthm}
\usepackage{enumitem}
\usepackage{tikz}
\usepackage{graphicx}
\usepackage{cite}
\newtheorem{lemma}{Lemma}
\newtheorem{corollary}{Corollary}

\newtheorem{theorem}{Theorem}

\usepackage[
all=normal,
paragraphs=tight,
floats=tight,
mathspacing=normal,
wordspacing=normal,
tracking=normal,
bibbreaks=tight
]{savetrees}
\usepackage{times}

\newcommand{\defn}{\emph}

\newcommand{\pred}{\hat{\sigma}}
\newcommand{\ind}{\text{index}}
\newcommand{\calT}{\mathcal{T}}
\newcommand{\polylog}{\text{polylog}}

\newcommand{\rebuild}{\text{rebuild}}

\title{Incremental Approximate Single-Source Shortest Paths with Predictions}
\author{
Samuel McCauley\thanks{Supported in part by NSF CCF 2103813}\\Williams College\\ \texttt{sam@cs.williams.edu}
\and 
Benjamin Moseley\thanks{Supported in part by Google Research Award, an Infor Research Award, a Carnegie Bosch Junior Faculty Chair, NSF grants CCF-2121744 and CCF-1845146.}\\ Carnegie Mellon University\\ \texttt{moseleyb@andrew.cmu.edu}  
\and 
Aidin Niaparast\thanks{Supported in part by the Air Force Office of
Scientific Research under award number FA9550-23-1-0031 and NSF award CCF-1845146.
} \\ Carnegie Mellon University\\ \texttt{aniapara@andrew.cmu.edu}
\and 
Helia Niaparast\thanks{ Supported in part by NSF CCF-1845146.
} \\ Carnegie Mellon University\\ \texttt{hniapara@andrew.cmu.edu}
\and
Shikha Singh\thanks{Supported in part by NSF CCF 1947789}\\Williams College\\  \texttt{shikha@cs.williams.edu}
}
\date{}

\begin{document}

\maketitle

\begin{abstract}
The algorithms-with-predictions framework has been used extensively to develop online algorithms with improved beyond-worst-case competitive ratios.  Recently, there is growing interest in leveraging predictions for designing data structures with improved beyond-worst-case running times.
In this paper, we study the fundamental data structure problem of maintaining approximate shortest paths in incremental graphs in the algorithms-with-predictions model.
Given a sequence $\sigma$ of edges that are inserted one at a time,  the goal is to maintain approximate shortest paths from the source to each vertex in the graph at each time step.   
Before any edges arrive, the data structure is given a prediction of the online edge sequence $\hat{\sigma}$ which is used to ``warm start'' its state.   

As our main result, we design a learned algorithm that maintains $(1+\epsilon)$-approximate single-source shortest paths, which runs in $\tilde{O}(m \eta \log W/\epsilon)$ time, where $W$ is the weight of the heaviest edge and $\eta$ is the prediction error.  We show these techniques immediately extend to the all-pairs shortest-path setting as well.
Our algorithms are consistent (performing nearly as fast as the offline algorithm) when predictions are nearly perfect, have a smooth degradation in performance with respect to the prediction error and, in the worst case, match the best offline algorithm up to logarithmic factors. That is, the algorithms are ``ideal'' in the algorithms-with-predictions model.

As a building block, we study the \emph{offline incremental} approximate single-source shortest-path (SSSP) problem.   In the offline incremental SSSP problem, the edge sequence $\sigma$ is known a priori and the goal is to construct a data structure that can efficiently return the length of the shortest paths in the intermediate graph $G_t$ consisting of the first $t$ edges, for all $t$.
Note that the offline incremental problem is defined in the worst-case setting (without predictions) and is of independent interest. 
\end{abstract}
\section{Introduction}
\label{sec:introduction}
The efficiency of algorithms is typically measured in the worst-case model.  The worst-case model makes a fundamental assumption that algorithmic problems are solved from scratch. In real applications, many problems are solved repeatedly on similar input instances.  There has been a growing interest in improving the running time of algorithms by leveraging similarities across problem instances.  The goal is to \emph{warm start} the algorithm; that is, initialize it with a machine-learned state, speeding up its running time.  This
initial state is learned from the past instances of the problem.   

This recent line of work has given rise to a widely-applicable model for beyond-worst-case running time analysis.   The area has come to be known as \defn{algorithms with predictions}. In the algorithms-with-predictions model, the goal is to establish strong worst-case-style guarantees for the algorithm.  The critical difference is that the running time of the algorithm is \emph{parameterized} by the quality of the learned starting state; that is, the prediction quality. The algorithms designed in this model are also frequently referred to as \emph{learning-augmented} or simply \emph{learned} algorithms.

Ideally, an algorithm outperforms the best worst-case algorithm with high quality predictions (i.e. the algorithm is \defn{consistent}). When the predictions are incorrect,  the algorithm's performance should degrade proportionally to the error (i.e. the algorithm should be \defn{smooth}) and never be worse than the best worst-case algorithm (i.e. the algorithm is \defn{robust}). If an algorithm is consistent, smooth, and robust, it is called an \defn{ideal learned algorithm}. 

Kraska et al.~\cite{KraskaBCDP18}  \emph{empirically} demonstrated how machine-learned predictions can speed up data structures.  
This seminal paper inspired Dinitz et al.~\cite{DinitzILMV21} to propose a \emph{theoretical} framework to leverage predictions to speed up \emph{offline} algorithms.
Since this work, several papers have used predictions to improve the running time of offline combinatorial optimization problems such as maximum flow~\cite{DaviesMVW,POLAK2024106487}, shortest path~\cite{LattanziSV23}, and convex optimization~\cite{SakaueO22}. Recently, Srinivas and Blum \cite{SrinivasB25} give a framework for utilizing multiple offline predictions to improve the running time of online optimization problems.
This growing body of theoretical and empirical work shows the incredible potential of using machine-learned predictions to improve the efficiency of algorithms. The prediction framework provides a rich and algorithmically interesting landscape that is yet to be understood. 

Data structures are a critical building block for dynamic algorithms for optimization problems.   
There is a growing body of theoretical work on designing ideal algorithms for data structure problems~\cite{LinLW22,
McCauleyMNS23,
McCauleyMoNi24,
BrandFNP24, liu2023predicted, 
benomar2024learningaugmented, 
zeynali2024robust, bai2023sorting,
dinitz2024binary}; see Section~\ref{sec:related} for details. 
Predictions have been particularly effective at speeding up 
data structures for dynamic graph problems, which typically incur polynomial update times in the worst case.  
In particular, McCauley et al.~\cite{McCauleyMoNi24} use
predictions to design faster data structures for incremental topological maintenance and cycle detection.  
Henzinger et al.~\cite{HenzingerSSY24} and van den Brand et al.~\cite{BrandFNP24} initiate the use of predictions for designing faster dynamic graph algorithms.
In particular, van den Brand~\cite{BrandFNP24} solve the online matrix-vector multiplication with predictions and use it to obtain faster algorithms for several dynamic graph problems such as incremental all-pairs shortest paths, reachability, and triangle detection. Liu and Srinivas~\cite{liu2023predicted} give efficient black-box transformations from  offline divide-and-conquer style algorithms to fully
dynamic learned algorithms that are given predictions about the update sequence.  
These initial results demonstrate incredible potential and there is a need to develop algorithmic techniques for designing data structures that can leverage predictions.  
Towards this goal, in this paper, we study the fundamental data structure problem of maintaining approximate single-source shortest paths in dynamic graphs with edge insertions. No learned data structure has been developed for this problem despite being a clear target in the area. 

\paragraph{Incremental Single-Source Shortest Paths.}  In this paper, we design a data structure to maintain shortest paths in a weighted directed graph when edges are inserted over time. 
Initially, all nodes $V$ of the graph are available and, in the single-source case, a source $s$ is specified. There are $m$ edges that arrive one by one:  at each time step $t$, an edge (with a positive weight) is added to the graph.   The goal is to design a data structure that approximately stores the shortest path distance from the source $s$ to every other vertex $v$ in the graph.   
Let $d^t(s,v)$ be the distance between $s$ and $v$ after $t$ edge insertions.  Let $\hat{d}^t(s,v)$ be an approximation of $d^t(s,v)$  computed by the algorithm after $t$ edges arrive. The algorithm needs to efficiently compute $\hat{d}^t(s,v)$ such that $d^t(s,v) \leq \hat{d}^t(s,v) \leq (1 + \epsilon) d^t(s,v)$ for some constant $\epsilon  > 0$.  

Incremental shortest paths is a fundamental algorithmic problem used in applications as well as a building block for other data structures~\cite{RodittyZ11,HenzingerKN16}. 
This problem has been extensively studied in the literature~\cite{GutenbergW20,BernsteinGS21,BernsteinGW20,KyngMeGu22,chechik2021incremental,henzinger2014sublinear, henzinger2015improved}. 
The best-known worst-case update times for the problem are $\tilde{O}(n
^2 \log W/\epsilon^{2.5} + m)$~\cite{probst2020new} for dense graphs and a $\tilde{O}(m^{4/3}\log W/\epsilon^2)$ algorithm~\cite{KyngMeGu22} for sparse graphs.\footnote{The $\tilde{O}$ notation suppresses log factors.}  

Recently, Henzinger et al.~\cite{HenzingerSSY24} and van den Brand~\cite{BrandFNP24} applied predictions to the problem of computing all-pairs shortest paths (APSP) in incremental graphs.  
We follow their prediction model in which the data structure is given a prediction of the online edge sequence $\hat{\sigma}$ before any edges arrive.  The performance of the algorithm is given in terms of the prediction error.  Define an edge $e$'s error $\eta_e$ as the difference between its arrival time in $\hat{\sigma}$ and $\sigma$ and the aggregate error $\eta$ as the $\max_{e} \eta_e$.  They show that for the incremental APSP problem, predictions can be used to support $O(1)$-lookup time and $O(\eta^2)$ time per edge insert, which is optimal under the Online Matrix-Vector Multiplication Hypothesis~\cite{BrandFNP24}.  The preprocessing time used in \cite{HenzingerSSY24} is $O(mn^3)$, and it is $O(n^{(3 + \omega)/2})$ in \cite{BrandFNP24} where $\omega$ is the exponent from matrix multiplication. 

Their work leaves open an important question---can predictions also help speed up the related and fundamental problem of maintaining approximate \emph{single-source} shortest paths under edge inserts, which has not yet been studied in this framework.
 
\paragraph{Our Contributions.}   The main contribution of this paper is a new learned data structure for the incremental approximate SSSP problem and a demonstration that it is ideal (consistent, robust, and smooth) with respect to the prediction error.
As a building block, we study the \emph{offline} version of the problem that has not previously been considered, which is of independent interest. 

We show that these techniques extend to the all-pairs shortest-path problem as well. 

\paragraph{Offline Incremental Single-Source Shortest Paths.}  We give a new algorithm for the \emph{offline} version of the incremental shortest path problem. In the offline version of the problem, the sequence of arriving edges is given in advance where each edge is assigned a unique time.  The goal is to maintain approximate distances $\hat{d}^t(s,v)$ for all $v$ and all times $t$ as efficiently as possible.  That is, given a query $(v,t)$, the data
structure outputs the approximate shortest path from source $s$ to
vertex $v$ in the graph with edges inserted up to time $t$.
By reversing time, the incremental and decremental versions of this problem are immediately equivalent.

Surprisingly, to our knowledge, past work in the offline setting has focused solely on exact, rather than approximate, incremental shortest path.  These exact versions have strong lower bounds. Roddity and Zwick~\cite{RodittyZ11} show that for the incremental/decremental single-source shortest paths problem in weighted directed (or undirected) graphs, the amortized query/update time must be $n^{1-o(1)}$, unless APSP can be solved in truly subcubic time (i.e. $n^{3-\epsilon}$ for constant $\epsilon>0$). 

We show that the offline \emph{approximate} version of the problem can be solved significantly faster than the exact version in the worst-case setting.  This natural problem reveals key algorithmic ideas for designing our learned online SSSP algorithm.

\begin{theorem}
\label{thm:offline-approx}
For the offline incremental SSSP problem there exists an algorithm running in worst-case total time $O(m \log (nW)  (\log^3 n)(\log \log n)/\epsilon)$ that returns $(1+\epsilon)$ approximate single-source shortest paths for each time $t$.
\end{theorem}

\paragraph{Predictions Model and Learned Online Single Source Shortest Paths.} 
Let $\sigma = e_1, \ldots, e_m$ denote the actual online sequence of edge inserts.  Before any edges arrive, the algorithm receives a prediction of this sequence, $\hat{\sigma}$.  This is the same prediction model considered by Henzinger et al.~\cite{HenzingerSSY24} and van den Brand~\cite{BrandFNP24} for the all-pairs shortest-paths problem. Let $\ind(e)$ be the index of $e$ in $\sigma$ and $\widehat{\ind}(e)$ be the index of $e$ in $\hat{\sigma}$.
Define $\widehat{\ind}(e):=m+1$ for edges $e$ that are not in $\hat{\sigma}$. 
Let $\eta_e = |\ind(e) - \widehat{\ind}(e)|$ for each edge $e$ in $\sigma$.

We first describe the performance of our learned SSSP algorithm in terms of parameters $\tau$ and $\text{HIGH}(\tau)$.  For any $\tau$, define $\text{HIGH}(\tau)$ to be the set of edges $e$ in $\sigma$ with error $\eta_e > \tau$. Then, we show that the bounds obtained are more robust than several natural measures of prediction error.  

\begin{theorem}\label{thm:online}
    There is a learned online single-source shortest path algorithm that given a prediction $\hat{\sigma}$ gives the following guarantees:
    \begin{itemize}[noitemsep, nolistsep]
        \item The algorithm maintains a $(1+\epsilon)$-approximate shortest path among edges that have arrived. 
        \item  The total running time for all edge inserts is
        $\tilde{O}(m \cdot \min_{\tau} \{\tau + |\text{HIGH}(\tau)| \} \log W/\epsilon)$.
    \end{itemize}
\end{theorem}

Our algorithm uses the offline algorithm as a black box, and therefore our results also apply to the decremental problem in which edges are deleted one by one.

This theorem can be used to give results for two natural error measures. 
We call the first error measure, the \defn{edit distance} $\text{Edit}(\sigma,\hat{\sigma})$, defined as the minimum number of insertions and deletions needed to transform $\sigma$ to the prediction $\hat{\sigma}$. To the best of our knowledge this is a new error measure.
  
The second error measure is  $ \eta = \max_{e \in \sigma} \eta_e $  the maximum error of any edge; this measure was also used by past work (e.g.~\cite{BrandFNP24, McCauleyMNS23, McCauleyMoNi24}).
The theorem gives the following corollary.

\begin{corollary}\label{cor:main}
There is a learned online algorithm that maintains $(1+\epsilon)$-approximate shortest paths and has running time at most the minimum of $\tilde{O}(m \cdot \text{Edit}(\sigma,\hat{\sigma}) \log W/\epsilon)$ and $\tilde{O}(m \eta \log W/\epsilon)$.
\end{corollary}
The corollary can be seen as follows. 
By definition, there are no edges in $\text{HIGH}(\eta)$.  Thus, setting $\tau = \eta$ in Theorem~\ref{thm:online}  gives an $\tilde{O}(m \eta \log W/\epsilon)$ running time.  Alternatively, setting $\tau = \text{Edit}(\sigma,\hat{\sigma})$  ensures that $\text{HIGH}(\tau)$
contains at most $\tau$ edges. This is because any edge not inserted or deleted in the process of transforming $\hat{\sigma}$ to $\sigma$ can move at most $\tau$ positions and so only inserted or deleted edges contribute to $\text{HIGH}(\tau)$.
This gives a running time of $\tilde{O}(m\ \text{Edit}(\sigma,\hat{\sigma}) \log W/\epsilon)$.

\paragraph{Discussion On Single-Source Shortest Paths.}   Notice that if a small number of edges have a large error $\eta_e$, the $\text{Edit}(\sigma,\hat{\sigma})$ is small even though the maximum error is large, and thus the algorithm retains strong running time guarantees on such inputs.   
Furthermore, $\text{Edit}(\sigma,\hat{\sigma})$ is small even if there is a small number of edges that are not predicted to arrive but do, or are predicted to arrive but never do (such edges are essentially insertions and deletions).  

On the other hand, the edit distance bound is a loose upper bound in some cases: for example, even if a large number of edges are incorrect, but have small relative change in position between $\sigma$ and $\hat{\sigma}$, then $\eta$ will be small. 

The algorithm is \emph{consistent} as its running time is optimal (since $\Omega(m)$ is required to read the input) up to log factors when the predictions are perfect.
It is \emph{smooth} in that it has a slow linear degradation of running time in terms of the prediction error.  We remark that \emph{robustness} to arbitrarily erroneous predictions can be achieved with respect to any worst-case algorithm simply by switching to the worst-case algorithm in the event that the learned algorithm's run time grows larger than the worst-case guarantee. 

\paragraph{Extension to All-Pairs  Shortest Paths.}  Next, we show that our techniques are generalizable by applying them to the all-pairs shortest-path (APSP) problem.  Similar to the SSSP case, we first solve the offline incremental version of the problem by running the SSSP algorithm multiple times.  We then extend it to the online setting with predictions. 

For the incremental APSP problem, it does not make sense to consider amortized cost---Bernstein~\cite{bernstein2016maintaining} gives an algorithm with nearly-optimal total work for the online problem even without predictions.  As a result, past work on approximate all-pairs shortest path with predictions has focused on improving the worst-case time for each update and query.  

For the APSP problem, we follow~\cite{HenzingerSSY24,BrandFNP24} and assume that $\hat{\sigma}$ is a permutation of $\sigma$---the set of edges predicted to arrive are exactly the set that truly arrive.\footnote{While this is in some cases a strong assumption, it seems unavoidable for worst-case update and query cost.} 

\begin{theorem}\label{thm:online-apsp}
    There is a learned online all-pairs shortest path algorithm with $\tilde{O}(nm\log W / \epsilon)$ preprocessing time, $O(\log n)$ worst-case update time, and $O(\eta^2\log\log_{1 + \epsilon} (nW))$ worst-case query time.
\end{theorem}

\paragraph{Comparison to Prior Work.} 

To the best of our knowledge, no prior work has considered the incremental SSSP problem in the algorithms-with-predictions framework. Our algorithm for the SSSP problem has a recursive tree of subproblems, that we bring online with predictions on the input sequence.  We remark that the work of   Liu and Srinivas \cite{liu2023predicted} gave a general framework for taking offline recursive tree algorithms into the online setting with predictions. Their framework is general and applies to a large class of problems that can be decomposed into recursive subproblems with smaller cost.
In contrast, our tree-based decomposition technique is tailored specifically to shortest paths which enables our efficient runtime.  Moreover, our analysis differs significantly from~\cite{liu2023predicted} as well---we cannot spread the cost evenly over smaller-cost recursive subproblems.  This is because a single edge insert can result in $\Omega(n)$ changes in shortest paths.  To avoid this, we allow subproblems to grow and shrink as necessary and charge the cost of a large subproblem to a large number of distance changes; see Section~\ref{sec:technical}.

The incremental APSP problem with predictions was studied by past work~\cite{HenzingerSSY24,BrandFNP24}.  Henzinger et al.~\cite{HenzingerSSY24} achieve $\tilde{O}(1)$ update time and $\tilde{O}(\eta^2)$ query time, after $O(mn^3)$ preprocessing for any weighted graph.
van den Brand et al.~\cite{BrandFNP24} achieve similar update and query times with $\tilde{O}(n^{(3 + \omega)/2})$ preprocessing time; however, the result is limited to unweighted graphs.  
They further show that this query time is essentially optimal: under the Online Matrix-Vector Multiplication Hypothesis, it is not possible to obtain $O(\eta^{2-\delta})$ query time for any $\delta > 0$ while maintaining polynomial preprocessing time and $O(n)$ update time.
Thus, Theorem~\ref{thm:online-apsp} obtains faster preprocessing time for sparse graphs and supports weighted graphs.  Finally, we note that Liu and Srinivas~\cite{liu2023predicted} give bounds for the \emph{fully dynamic} weighted APSP problem in the prediction-deletion model they propose and their bounds are incomparable to ours.
\subsection{Additional Related Work}\label{sec:related}

\paragraph{Data Structures with Predictions.}

Data structures augmented with predictions have demonstrated empirical success in key applications such as indexing \cite{KraskaBCDP18,ding2020alex,DaiXGA20}, caching \cite{JiangP020,LykourisVassilvitskii21}, 
Bloom filters \cite{mitzenmacher2018model,vaidya2020partitioned}, frequency estimation \cite{hsu2019learning}, page migration \cite{IndykMMR22}, routing \cite{BhaskaraGKM20}.  

Initial theoretical work on data structures with predictions focused on improving their space complexity or competitive ratio, e.g. learned filters~\cite{vaidya2020partitioned,mitzenmacher2018model,bercea2022daisy} and count-min sketch~\cite{hsu2019learning,du2021putting} on stochastic learned distributions.  Several papers have since used predictions to provably improve the running time of dictionary data structures, e.g.\ learned binary-search trees~\cite{lin2022learning, chen2022power, zeynalirobust,dinitz2024binary}.  McCauley et al.~\cite{McCauleyMNS23} presented the first ``ideal'' data structures with prediction for the list-labeling problem~\cite{McCauleyMNS23}. They use a ``prediction-decomposition'' framework to obtain their bounds and extend the technique to maintain topological ordering and perform cycle detection in incremental graphs with predictions~\cite{McCauleyMoNi24}.

\paragraph{Incremental SSSP and APSP in the Worst-Case Setting.} 
We review the state-of-the-art deterministic worst-case algorithms for the incremental $(1+\epsilon)$ SSSP problem.  

The best-known deterministic algorithm for dense graphs is by Gutenberg et al.~\cite{probst2020new} with total update time $\tilde{O}(n^2 \log W/\epsilon^{O(1)})$, which is nearly optimal for very dense graphs.  For sparse graphs,
Chechik and Zhang~\cite{chechik2021incremental} give an algorithm with total time $\tilde{O}(m^{5/3}\log W/\epsilon)$, which is improved by 
Kyng et al.~\cite{KyngMeGu22} to a total time $\tilde{O}(m^{3/2}\log W /\epsilon)$; this is the best-known deterministic algorithm for sparse graphs.  Note that many of these solutions use the ES-tree data structure~\cite{shiloach1981line} as a key building block. The ES-tree can maintain exact distances in an SSSP tree for weighted graphs with total running time $O(mnW)$~\cite{henzinger1995fully}, where $W$ is the ratio of the largest to smallest edge weight. The ES-tree can be used to maintain $(1+\epsilon)$-approximate distances in total update time $\tilde{O}(mn\log W/\epsilon)$ for incremental/decremental directed SSSP; see e.g.~\cite{bernstein2009fully,bernstein2016maintaining,madry2010faster}.

This paper shows that even modestly accurate predictions can be leveraged to circumvent these high (polynomial) worst-case update costs in incremental SSSP.

For the approximate incremental APSP problem without predictions, Bernstein~\cite{bernstein2016maintaining} gave an algorithm that has nearly-optimal total runtime $O(nm \log^4 (n) \log (nW)/\epsilon)$ and $O(1)$ look-up time. 

Peng and Rubinstein consider generally how incremental or decremental algorithms can be turned into fully dynamic algorithms \cite{PengR23}.
\section{Preliminaries}
\label{sec:prelim}
Let $\sigma = e_1, \ldots, e_m$ be the sequence of all edge insertions, and let $V$ be the set of vertices. 
We assume $m$ is of the form $2^k$ throughout this paper, which can be assumed without loss of generality by allowing dummy edge inserts.
Each edge $e_i$ has a weight $w(e_i) \in [1, W]$.
Let $G_t$ be the graph consisting of vertices $V$ and the first $t$ edges $e_1, \ldots, e_t$; we call this the \emph{graph at time $t$}. 
We define $G_0$ to be the empty graph on the vertex set $V$.

The \emph{length} of a path is the sum of the weights of its edges. 
Let $d^t(v)$ be the length of the shortest path from $s$ to $v$ in $G_t$. 
Throughout this paper, all logarithms are base 2 unless otherwise specified.

\paragraph{Problem Definition.}
In the offline problem, the algorithm is given a sequence of edge inserts $\sigma = e_1, e_2, \ldots, e_m$, the set of vertices $V$, and the source vertex $s$.  The goal of the algorithm is to output a data structure that answers queries $(v,t)$ of the form: what is the length of the shortest path from $s$ to $v$ at time $t$?  This answer should be a $(1 + \epsilon)$-approximation: specifically, for a query $(v, t)$, if the data structure returns $d$, then ${d}^t(v) \leq d \leq d^t(v)(1 + \epsilon)$.

In the online problem, the edges in $\sigma = e_1, \ldots, e_m$ arrive one at a time.  Before any edge arrives, the algorithm is given a prediction $\hat{\sigma}$ of $\sigma$, as well as the set of vertices $V=\{v_1,\ldots,v_n\}$ and the source vertex $s$.  
We assume the length of $\hat{\sigma}$ is the same as $\sigma$.
At all times, the algorithm must maintain an array $D$ containing the length of the shortest path to each vertex.  In particular, after edge $t$ is inserted, $D$ must satisfy $d^t(v_i)\leq D[i] \leq d^t(v_i)(1 + \epsilon)$ for all vertices $v_i$.

\paragraph{Adjusting $\epsilon$.}
Throughout this paper, we assume $\epsilon = O(1)$.
In our algorithms, distance estimates $\hat{d}^t(v)$ are used for each node $v$ and time $t$ that satisfy the following invariant (see Lemmas~\ref{lem:refined_approx} and~\ref{lem:pred_approx}):
$d^t(v) \leq \hat{d}^t(v) \leq d^t(v)(1 + \epsilon/ \log m)^{\log m}$.

We will need to slightly adjust $\epsilon$ in our algorithms to account for lower-order terms.
For $\epsilon<1.79$, it is known that
\[\left( 1 + \frac{\epsilon}{\log m}\right)^{\log m} \leq e^\epsilon \leq 1+\epsilon+\epsilon^2,\]
which ensures that $d^t(v) \leq \hat{d}^t(v) \leq d^t(v)(1 + \epsilon + \epsilon^2)$.
To ensure our algorithms return $(1+\epsilon)$ approximations for $d^t(v)$, the $\epsilon$ parameters used are set to be $(\min\{1.79,\epsilon\})/4$. 

\paragraph{Defining Recursive Subproblems.}
The offline algorithm is recursive.
It divides the interval $[0,m]$ in half and recurses on both sides. On each recursive call, the algorithm will process an interval $[\ell, r]$ and further subdivide the interval until it contains a single edge. 
We call each recursive call $[\ell,r]$ a \emph{subproblem}.
We refer to this subproblem both as \emph{subproblem $[\ell,r]$} and \emph{subproblem $x$}, where $x=(\ell+r)/2$ is the midpoint of $[\ell,r]$.
Each subproblem corresponds to a node in the recursion tree.
We might refer to the subproblem $[\ell,r]$ in the recursion tree as \emph{node $x$} in the recursion tree.
Each time $x \in \{1,\ldots,m-1\}$ is the midpoint of exactly one node in the tree.
Let the \defn{level} of a time $x$ be the depth of node $x$ in the recursion tree. 
So the level of $m/2$ is $1$, the level of $m/4$ and $3m/4$ is $2$, and so on.
In particular, the level of $x$ is one more than the maximum level of $\ell$ and $r$.
For ease of notation, we set the level of 0 and $m$ to be 0. 
See Figure~\ref{fig:tree} for an illustration of the recursion tree.
We refer to the ancestors and descendants of a node in the recursion tree as \emph{ancestors} and \emph{descendants} of the corresponding subproblem.
All these definitions extend to the online algorithm as well, as it is built up on the offline algorithm.

\begin{figure}
    \centering
    \begin{tikzpicture}[main/.style = {fill = black, circle, inner sep = 2pt}, 
    label/.style = {fill = none, circle, inner sep = 3pt}, 
    scale = 0.9]
        \draw[step = 1cm, gray, thin] (0,1) grid (16,5);
            \filldraw[fill = white, draw = black] (0,5) rectangle (16,4);
        \foreach \i in {0, 8}
            \filldraw[fill = white, draw = black] (\i,4) rectangle (\i+8,3); 
        \foreach \i in {0, 4, ..., 12}
            \filldraw[fill = white, draw = black] (\i,3) rectangle (\i+4,2);    
        \foreach \i in {0, 2, ..., 14}
            \filldraw[fill = white, draw = black] (\i,2) rectangle (\i+2,1);

        \foreach \i in {0, 1, 2, ..., 16}
            \node[label] at (\i, 0.5) {$G_{\i}$};
        
        \node[main] at (8, 4.5) {};
        \foreach \i in {4, 12}
            \node[main] at (\i, 3.5) {};
        \foreach \i in {2, 6, 10, 14}
            \node[main] at (\i, 2.5) {};
        \foreach \i in {1, 3, ..., 16}
            \node[main] at (\i, 1.5) {};
    
        \foreach \i in {4, 12}
           \draw (8, 4.5) -- (\i, 3.5);
        \foreach \i in {2, 6}
           \draw (4, 3.5) -- (\i, 2.5);
        \foreach \i in {10, 14}
           \draw (12, 3.5) -- (\i, 2.5);
        \foreach \i in {1, 3}
           \draw (2, 2.5) -- (\i, 1.5);
        \foreach \i in {5, 7}
           \draw (6, 2.5) -- (\i, 1.5);
        \foreach \i in {9, 11}
           \draw (10, 2.5) -- (\i, 1.5);
        \foreach \i in {13, 15}
           \draw (14, 2.5) -- (\i, 1.5);  
    \end{tikzpicture}    
    \vspace{-.15in}
    \caption{The recursion tree of the algorithm. Each node $x$ in the tree is associated with an interval $[\ell,r]$ such that $x=(\ell+r)/2$. The depth of a node is the number of nodes in the path from the root $m/2$ to that node. For example, the depth of node $x=6$ is 3, and its corresponding interval is $[4,8]$.}
    \label{fig:tree}
    \vspace{-.15in}
\end{figure}
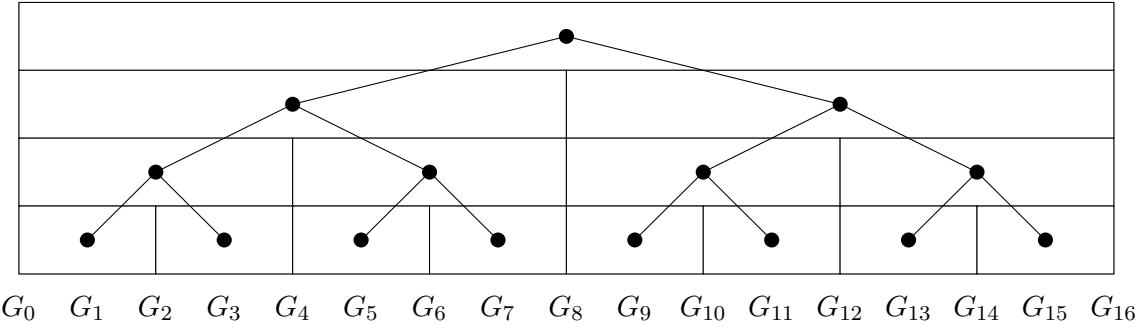

\section{Technical Overview}
\label{sec:technical}

\subsection{Offline Algorithm}
The preliminary observation behind our algorithm is the following. Let's say we round the distance to each vertex up to the nearest power of $(1 + \epsilon)$.  Since the maximum distance to any vertex is $nW$, this means that we group distances into $O(\log_{1 + \epsilon} (nW)) = O(\log (nW)/\epsilon)$ ``buckets.''  Knowing that the distance to any vertex only decreases over time, the buckets of all vertices only change $O(n\log (nW)/\epsilon)$ times in total.  The goal of our algorithm is to list the $O(\log (nW)/\epsilon)$ times when a vertex shifts from one bucket to the next.

Thus, when an edge is inserted at time $t$, our goal is to run Dijkstra's only on vertices whose distance is changing at time $t$, charging the cost to the changing distances.  

\paragraph{Removing Vertices from Recursive Calls.} 
Our algorithm recursively divides the sequence of edge inserts in half, starting with the entire sequence of edge inserts $\{e_1, \ldots, e_m\}$, and recursing until the sequence contains a single edge.   

Consider a sequence $\{e_{\ell}, \ldots, e_r\}$, which we denote by $[\ell, r]$.  We divide this interval into two equal halves: $[\ell, x], [x, r]$.  

The observation behind our algorithm is the following.  Let's say we calculate the distance to all vertices at time $x$.  Consider a vertex $v$ that is in the same bucket at $\ell$ and $x$.
Since the distance to $v$ is non-increasing as new edges arrive, $v$ must be in this bucket throughout the interval, and we do not need to continue calculating its distance. 
Therefore, we should remove $v$ from the recursive call for the interval $[\ell, x]$.  
Similarly, if $v$ is in the same bucket at $x$ and $r$, then we should remove $v$ from $[x, r]$.

If we are successful in removing vertices this way, we immediately improve the running time.  Each vertex $v$ is included in an interval $[\ell, r]$ if and only if it shifts from one bucket to another between $\ell$ and $r$.  Thus, each time $v$ shifts buckets, it is included in at most $\log m$ intervals.  This means that each vertex is included in $O(\log m \log (nW)/\epsilon)$ intervals.  
 
Let's assume that for an interval $[\ell, r]$, we can run Dijkstra's algorithm for the midpoint $x$ in time (ignoring log factors) proportional to the number of edges that have at least one endpoint included in the interval $[\ell, r]$. 
Since each vertex contributes to the running time of $\polylog(nmW)/\epsilon$ different subproblems, we can sum to obtain $\tilde{O}(m\log W/\epsilon)$ total cost to calculate the shortest path to every vertex in every subproblem.

\paragraph{Ensuring Dijkstra's Runs are Efficient.}
Let us discuss in more detail how to ``remove'' a vertex from a recursive call.
We do not remove vertices from the graph; instead, we say that a vertex is \defn{dead} during an interval $[\ell, r]$ if its distance bucket does not change during the interval, and \defn{alive} otherwise.
Each edge $(u,v)$ is alive if and only if $v$ is alive.
  
For each time $x$  which is the midpoint of an interval $[\ell, r]$, consider building a graph $G'_x$
consisting of all alive vertices and edges in $[\ell, r]$.  If $(u,v)$ is alive, but $u$ is dead, we also add $u$ to $G'_x$.
Now, the time required to run Dijkstra's algorithm on $G'_x$ is roughly proportional to the number of edges in $G'_x$ as we desired, but the distances it gives are not meaningful---in fact, $s$ may not even be in $G'_x$.

To solve this problem, we observe that we already know (approximately) the distance to any dead vertex $u$: we marked $u$ as dead because its distance bucket does not change from $\ell$ to $r$.  Thus, we add $s$ to $G'_x$; then, for each alive edge $(u,v)$ where $u$ is dead, instead of adding $u$ to $G'_x$, we add an edge from $s$ to $v$ with weight equal to the rounded distance from $s$ to $u$ from the previous recursive call plus $w(u,v)$.
Running Dijkstra's algorithm on $G'_x$ now gives (ostensibly) meaningful distances from $s$.

We calculate the distance to each vertex in $G'_x$, round the distance up to obtain its bucket, and then recurse (marking vertices dead as appropriate) on $[\ell, x]$ and $[x, r]$.

\paragraph{Handling Error Accumulations.}

Unfortunately, the above method does not quite work, as the error accumulates during each recursive call.  

For an interval $[\ell, r]$ with midpoint $x$, we use the terms \emph{recursive call $[\ell, r]$} and \emph{recursive call $x$} interchangeably.
For a recursive call $x$, the shortest path to $v$ may begin with an edge $(s, v')$ (for some dead vertex $v'$) weighted according to the bucket of $v'$.  The bucket of $v'$ was calculated by rounding up the length of the shortest path to $v'$ in some graph $G'_{y}$ at time $y$; the length of this path again was also rounded up, and so on.

In other words, since the weight of each edge from $s$ to a dead vertex $v'$ is based on the bucket of $v'$, it has been rounded up to the nearest power of $1 + \epsilon$.  Thus, \emph{each} recursive call loses a $1 + \epsilon$ approximation factor in distance.

To solve this, we note that the recursion depth of our algorithm is $\log_2 m$.  Thus, we split into finer buckets: rather than rounding up to the nearest power of $1 + \epsilon$, we round up to the nearest power of $1 + \epsilon/\log_2 m$.  Thus, the total error accumulated over all recursive calls is $(1 + \epsilon/\log_2 m)^{\log_2 m} \leq 1 + \epsilon + \epsilon^2$ (see Section~\ref{sec:prelim}); decreasing $\epsilon$ slightly gives error $1 + \epsilon$.  

Rounding buckets up to the nearest power of $1 + \epsilon/\log m$ changes our running time analysis: most critically, each vertex now changes buckets $\log_{1 + \epsilon/\log m} nW = \Theta(\log(nW) \log m/\epsilon)$ times.  

\subsection{Ideal Algorithm with Predictions}

With perfect predictions, the algorithm described above runs in $\tilde{O}(m\log W/\epsilon)$ time: we can just run the algorithm on the predicted (in fact the actual) sequence $\hat{\sigma} = \sigma$ to obtain an approximate distance to every vertex at each point in time.  
We show how to carefully rebuild portions of the offline approach to obtain an ideal algorithm.

\paragraph{Warmup: Hamming Distance of Error.}  To start, let's give a simple algorithm that uses predictions that contain error.  
Let $Ham(\sigma, \hat{\sigma})$ be the Hamming distance between $\sigma$ and $\hat{\sigma}$---in other words, the number of edges whose predicted arrival time was not exactly correct in $\hat{\sigma}$.  We give an algorithm that runs in $\tilde{O}(m\cdot Ham(\sigma, \hat{\sigma})\log W/\epsilon)$ time.

The main idea behind this algorithm is to \emph{update} the sequence of predictions $\hat{\sigma}$ as edges come in.  At time $t$, an edge $e$ arrives; if $e$ arrived at $t$ in $\hat{\sigma}$, the algorithm does nothing---the predictions created by running the offline algorithm on $\hat{\sigma}$ took $e$ arriving at $t$ into account.  If $e$ did not arrive at $t$ in $\hat{\sigma}$, then the algorithm creates a new $\hat{\sigma}$, replacing the edge arriving at $t$ with $e$.  The algorithm then reruns the offline algorithm on the new $\hat{\sigma}$ in $\tilde{O}(m\log W/\epsilon)$ time.

Since we run the offline algorithm (in $\tilde{O}(m\log W/\epsilon)$ time) each time an edge was predicted incorrectly, we immediately obtain $\tilde{O}(m\cdot Ham(\sigma, \hat{\sigma})\log W/\epsilon)$ time.

\paragraph{Handling Nearby Edges More Effectively.}
Ideally, we would not rebuild from scratch each time an edge is predicted incorrectly---we would like the running time to be proportional to \emph{how far} an edge's true arrival time is from its predicted arrival time.  

Our final algorithm improves over the warmup Hamming distance algorithm in two ways.  First, it updates the predicted sequence $\hat{\sigma}$ more carefully.  Second, it only rebuilds parts of the recursive calls of the offline algorithm: specifically, only intervals that changed as $\hat{\sigma}$ was updated.

First, let's describe how to update $\hat{\sigma}$.  As before, when an edge $e$ arrives at time $t$, if $e$ was predicted to arrive at $t$ in $\hat{\sigma}$ the algorithm does nothing.  
If $e$ was predicted to arrive at time $t'$ in $\hat{\sigma}$, the algorithm modifies $\hat{\sigma}$ by inserting $e$ at time $t$ (shifting all edges after $t$ down one slot), and deleting $e$ at time $t'$ (shifting all edges after $t'$ up one slot).
If $e$ is not predicted in $\hat{\sigma}$, the algorithm only inserts $e$ at time $t$, shifting subsequent edges down one slot.

The only entries in $\hat{\sigma}$ that change are those between $t$ and $t'$ (inclusive).  Therefore, we only need to recalculate $G'$ for times between $t$ and $t'$.  

Finally, we update the distance array $D$.  The algorithm greedily updates $D$ during the above rebuilds to maintain the invariant that $D[i]$ stores the estimated distance $\hat{d}^t(v_i)$, which as discussed in Section~\ref{sec:prelim}, is a $(1+\epsilon)$ approximation for $d^t(v_i)$.

\paragraph{Analysis.}
For any edge $e$ that appears at time $\ind(e)$ in $\sigma$ and time $\widehat{\ind}(e)$ in $\hat{\sigma}$, we define $\eta_e = |\ind(e) - \widehat{\ind}(e)|$.  If $e$ does not appear in $\hat{\sigma}$ then $\eta_e = |m+1-\ind(e)|$.
Let $\eta$ be the maximum error: $\eta = \max_{e \in \sigma} \eta_e$.
 
We show that an edge $e$ causes a rebuild of a graph $G'_t$ only if $t$ is between its predicted arrival time $\widehat{\ind}(e)$ and its actual arrival time $\ind(e)$. 
This means that each $G'_t$ can only be rebuilt $O(\eta)$ times.  Since the total time to build all $G'_t$ is $\tilde{O}(m \log W/\epsilon)$, the total time to rebuild all $G'_t$ $\eta$ times is $\tilde{O}(m\eta\log W/\epsilon)$.

With a more careful analysis, we can get the best of the Hamming analysis and the max error analysis.  For any $\tau$, let $\text{HIGH}(\tau)$ be the set of edges with error more than $\tau$.  For all edges with error more than $\tau$, we may (in the worst case) rebuild the entire interval $\{1, \ldots, m\}$, for total cost $\tilde{O}(m |\text{HIGH}(\tau)| \log W/\epsilon)$.  For all edges with error at most $\tau$, the total rebuild cost is, as above, $\tilde{O}(m\tau\log W/\epsilon)$.  Thus, we obtain total cost $ \tilde{O}\left(m (\log W/\epsilon) \cdot \min_{\tau} \{ \tau+|\text{HIGH}(\tau)| \}\right)$.
\section{Offline Incremental Weighted Directed Shortest Path}
\label{sec:offline}
This section presents an algorithm for the offline problem which takes ${O(m\log (nW) (\log^3 n)\log\log n/\epsilon)}$ time to build, and $O(\log \log_{1+\epsilon}(nW))$ time to answer a query $(v,t)$.  

The algorithm maintains an estimate of the shortest path at all points in time.  
Let $\hat{d}^t(v)$ be the estimate of $d^t(v)$ obtained by the algorithm. 
The algorithm does not explicitly maintain $\hat{d}^t(v)$ values for each vertex $v$ and each time $t$, because of the following simple observation: 
if for times $\ell$ and $r$, $\ell<r$, we have $\hat{d}^\ell(v)=\hat{d}^r(v)$, it means that from the algorithm's perspective, node $v$ has the same distance from $s$ in graphs $G_{\ell}$ and $G_r$. 
Since the distances of the nodes from $s$
are non-increasing over time, the algorithm infers that $\hat{d}^t(v)=\hat{d}^\ell(v)$ for each $\ell \leq t \leq r$.
Although in such case $\hat{d}^t(v)$ is not explicitly stored by the algorithm, we still use this notation to refer to $\hat{d}^{\ell}(v)=\hat{d}^r(v)$.

In each subproblem $x$ with an interval $[\ell,r]$, where $x=(\ell+r)/2$, vertices and edges are marked by the algorithm as alive or dead.
A vertex is \emph{alive} in subproblem $x$ if its estimated distances are not the same in subproblems $\ell$ and $r$, otherwise it is \emph{dead}.
In each subproblem $x$ (i.e., node $x$ in the recursion tree), the distance estimates $\hat{d}^x(v)$ for all alive nodes $v$ are explicitly maintained in a balanced binary search tree.
This allows us to access $\hat{d}^x(v)$ in $O(\log n)$ time for each alive node $v$ in subproblem $x$.

Moreover, for each node $v$, the algorithm maintains a list $L_v$ of length
$\log_{1 + \epsilon}( nW)$, representing the times when $v$'s estimated distance from $s$ moves from one integer power of $(1+\epsilon)$ to another.
In particular, the $i$th entry for a vertex $v$, denoted by $L_v(i)$, represents the minimum $t$ such that $\hat{d}^t(v) \leq (1 + \epsilon)^i$.
This is the data structure that the algorithm outputs in order to answer queries $(v,t)$.
To answer a query $(v,t)$, i.e., to obtain a $(1+\epsilon)$-approximation of $d^t(v)$, the algorithm performs  a binary search on $L(v)$ to find an $i$ such that $L_v(i) \leq t < L_v(i-1)$, and then we return $d=(1+\epsilon)^i$.
The time needed to obtain $d$ is then $\log (\log_{1+\epsilon}(nW))$. 
Note that $\hat{d}^t(v) \leq d \leq (1+\epsilon)\hat{d}^t(v)$, which as discussed in Section~\ref{sec:prelim}, it means that $d$ is within a $(1+\epsilon)$ factor of $d^t(v)$ for the original $\epsilon$.

Now we are ready to define the algorithm.

\subsection{The Offline Algorithm}
\label{sec:offline_algorithm}

The offline algorithm is recursive.
It starts with the interval $[0,m]$, and it recursively divides the interval in half and continues on the two subintervals.
Initially, the algorithm marks all vertices and edges as alive in $G_m$, and all vertices as alive in $G_0$. The algorithm runs  Dijkstra's  on $G_m$, and for each vertex $v$, it sets $\hat{d}^m(v) = d^m(v)$. The algorithm sets $\hat{d}^0(s) = 0$, and for all $v \in V \setminus \{s\}$, it sets $\hat{d}^0(v) = \infty$.  The algorithm then begins the recursive process on the interval $[0, m]$.
The algorithm stops recursing when the interval $[\ell,r]$ contains only one new edge; that is, when $r-\ell=2$.

Consider when the algorithm is given an interval $[\ell, r]$ to process. 
On this recursive call (subproblem), the goal is to calculate the distance estimates $\hat{d}^x(v)$ for alive nodes $v$, where $x = (\ell + r)/2$ is the midpoint of $[\ell,r]$.
Since the algorithm processes the subproblems in the recursion tree from top to bottom, it has already processed the subproblems $\ell$ and $r$, i.e., the distance estimates are calculated for the alive nodes in $G_\ell$ and $G_r$.

Just to repeat, a vertex $v$ is alive in $G_x$ if it is alive in $G_\ell$ and $G_r$, and its estimated distance is \emph{not} the same in $G_\ell$ and $G_r$, i.e., $\hat{d}^\ell(v) \neq \hat{d}^r(v)$. If a vertex is not alive, it is dead.   
A directed edge $e = (u,v)$ in $G_x$ is said to be \emph{alive} if $v$ is alive in $G_x$; otherwise, $e$ is \emph{dead}. After the description of the algorithm, we discuss how to efficiently maintain alive and dead vertices and edges. 
Although the algorithm does not store $\hat{d}^x(v)$ for the dead vertices $v$ in the subproblem $x$, we still use the notation $\hat{d}^x(v)$ to refer to $\hat{d}^\ell(v) = \hat{d}^r(v)$.

Now, we can define how the algorithm calculates $\hat{d}^x(v)$ for all alive vertices $v$. 
The algorithm creates a \emph{new graph} $G'_x$ whose vertex set is only the alive vertices in $G_x$ along with $s$. Then, for each alive edge $e\in G_x$ with $e = (u,v)$, there are two cases:
\begin{itemize}
    \item If both $u$ and $v$ are alive,  add $e$ to $G'_x$.
    \item If only $v$ is alive, add an edge from $s$ to $v$ with weight $\hat{d}^x(u) + w(u,v)$ to $G'_x$.
\end{itemize}
To compute $\hat{d}^x(u)$, where $u$ is a dead vertex in $G_x$, the algorithm needs to find the first ancestor of the current subproblem in the recursion tree in which $u$ is alive.
To do this, the algorithm does a binary search on the ancestors of node $x$ in the recursion tree, and for the first ancestor $x'$ where $u$ is alive in $G_{x'}$, it recovers the value of $\hat{d}^{x'}(u)=\hat{d}^x(u)$ using the binary search tree stored at node $x'$. 

The algorithm then computes the distance from $s$ to each vertex in $G'_x$ using Dijkstra's algorithm.  For any alive vertex $v$, the algorithm stores the length of the shortest path to $v$ found by Dijkstra's algorithm rounded up to the nearest integer power of $(1+\epsilon/\log m)$ as $\hat{d}^x(v)$.  

\paragraph{Efficiently Maintaining Alive Edges.}
For any $x$, let $m_x$ be the number of alive edges in $G_x$.   The algorithm maintains
a list of all alive edges at each time $x$.  

We now describe how to find the alive vertices and edges in $G_x$, where $x$ is the midpoint of $[\ell, r]$.
Note that if an edge $(u,v)$ is alive in $G_x$, then its head $v$ must be alive in $G_x$, which in turn implies that $v$ is alive in $G_\ell$ and $G_r$. Since $(u,v)$ has arrived before time $r$, it is alive in $G_r$.
Therefore the alive edges in $G_x$ are a subset of the alive edges in $G_r$.
To obtain the list of alive edges in $G_x$, the algorithm iterates over the alive edges in $G_r$ in $O(m_r)$ time, and removes the edges that either arrived after time $x$, or whose head node has the same estimated distance in both $G_\ell$ and $G_r$.   
While doing this, the algorithm additionally updates if each vertex is alive or not in $G_x$ in $O(m_r)$ time.

From this, the algorithm constructs $G'_x$ in $O(m_r)$ time, and runs Dijkstra's algorithm in $O(m_r\log n)$ time. 

\paragraph{Efficiently Maintaining $L_v$ Lists of Time Indexed Distances.}
To obtain the $L_v(i)$ values, initially, all the entries of $L_v$ are empty. 
At each time $x$, when we are calculating the distance estimate $\hat{d}^x(v)$ for an alive node $v$, we update $L_v$: suppose $(1+\epsilon)^{i-1} < \hat{d}^x(v) \leq (1+\epsilon)^{i}$.
If $L_v(i)$ is empty or $x<L_v(i)$, we set $L_v(i)$ to be $x$. Otherwise, $L_v(i)$ remains unchanged. In the end, the algorithm processes each of the lists, and for any empty $L_v(i)$, the algorithm sets it to be equal to the last non-empty entry of $L_v$ before $L_v(i)$.

\subsection{Analysis of the Offline Algorithm}
\label{sec:offline_analysis}

This section establishes the correctness and running time guarantees of the algorithm.   

\begin{lemma}
\label{lem:refined_approx}
   If $x$ is at level $i$, then for any vertex $v$, $d^x(v) \leq \hat{d}^x(v) \leq d^x(v)(1 + \epsilon/ \log m)^i$. 
\end{lemma}
\begin{proof}
    We show the lemma by induction on $i$.  For $i = 0$, i.e., $x=0$ or $x=m$, we obtain the exact distance for every vertex. 
    
    First, we show that $d^x(v) \leq \hat{d}^x(v)$. If $v$ is dead in subproblem $x$, then $\hat{d}^x(v) = \hat{d}^\ell(v)$, where $\ell$ is the start of the interval of node $x$ in the recursion tree. So, the level of $\ell$ is smaller than $i$. By induction, $d^\ell(v) \leq \hat{d}^\ell(v)$. Also, since $\ell \leq x$, we have $d^x(v) \leq d^\ell(v)$. Therefore, 
    $d^x(v) \leq d^\ell(v) \leq \hat{d}^\ell(v) = \hat{d}^x(v).$ So, we can assume $v$ is alive.  
    Let $p = s, v_1, v_2, \cdots, v_k$ be the shortest path from $s$ to $v=v_k$ in $G'_x$, so $\hat{d}^x(v)$ is the length of $p$ rounded up to the nearest integer power of $1 + \epsilon/\log m$.  
    Note that the path $v_1, v_2, \cdots, v_k$ exists in $G_x$ as well, since all the edges that only appear in $G'_x$ start at $s$. 
    If the edge $(s,v_1)$ is present in $G_x$, then $d^x(v) \leq \hat{d}^x(v)$, as $p$ is a path from $s$ to $v$ in $G_x$, which means that $d^x(v)$, the length of the shortest path in $G_x$ from $s$ to $v$, is at most the length of $p$.
    Otherwise, there exists an edge $(u, v_1)$ in $G_x$, such that $u$ is dead at time $x$, and the edge $(s,v_1)$ has weight $\hat{d}^{x}(u) + w(u, v_1)$.
    As shown above, since $u$ is dead at time $x$, we have $\hat{d}^{x}(u) \geq d^x(u)$.
    Therefore, 
    \[\hat{d}^x(v) \geq \hat{d}^{x}(u) + w(u, v_1) + d_{G'_x}(v_1, v) \geq d^x(u) + w(u, v_1) + d_{G_x}(v_1, v) \geq d^x(v).\]
    
    Next, we show that $\hat{d}^x(v) \leq d^x(v)(1 + \epsilon/ \log m)^i$. 
    Let $q$ be the shortest path from $s$ to $v$ in $G_x$.
    If $q$ is also present in $G'_x$, it follows that $\hat{d}^x(v) \leq d^x(v)(1 + \epsilon/ \log m)$.
    Otherwise, let $(a,b)$ be the last edge in $q$ that does not exist in $G'_x$. 
    Thus, either $a$ or $b$ is not alive in $G_x$. 
    
    First, let's assume $b$ is not alive in $G_x$.  Thus $b$ must be the last vertex in $q$---if there were an edge $(b,c)$ in $q$, this edge would exist in $G'_x$, so $b$ would be alive in $G_x$ (note that $b \neq s$).
    Since $b$ is the last vertex in $q$, and $v=b$ is dead in $G_x$, it follows that 
    \[
        \hat{d}^x(v) = \hat{d}^r(v) 
        \leq d^r(v)(1 + \epsilon/ \log m)^{i-1} 
        \leq d^x(v)(1 + \epsilon/ \log m)^{i-1},
    \]
    where $r$ is the end of the interval of node $x$ in the recursion tree.
    The first inequality follows from the induction hypothesis and the fact that the level of $r$ is strictly less than $i$, and the second inequality is because $x \leq r$ and adding more edges can only decrease the distance of each node from $s$.
    
    Otherwise, assume $b$ is alive in $G_x$, which means that $a$ is dead in $G_x$.
    Thus, there is an edge $(s,b)$ in $G'_x$ with weight $\hat{d}^x(a)+w(a,b)=\hat{d}^{r}(a) + w(a,b)$.
    Prepending this edge to the suffix of $q$ beginning at $b$ results in a path in $G'_x$ from $s$ to $v$ of length $\hat{d}^{r}(a) + w(a,b) + d_{G_x}(b,v)$. 
    Since the level of $r$ is less than $x$, 
    by the induction hypothesis we have $\hat{d}^{r}(a) \leq d^{r}(a)(1+\epsilon/\log m)^{i-1}$. 
    Hence, 
    \begin{align*}
      \hat{d}^x(v) &\leq  \left(\hat{d}^{r}(a) + w(a, b) + d_{G_x}(b, v)\right)(1 + \epsilon/\log m) \\ 
      &\leq \left(d^r(a)(1+\epsilon/\log m)^{i-1}+ w(a, b) + d_{G_x}(b, v)\right)(1 + \epsilon/\log m) \\
      &\leq \left(d^x(a) + w(a,b) + d_{G_x}(b, v)\right)(1 + \epsilon/\log m)^i \\
      &= d^x(v)(1 + \epsilon/\log m)^i. \qedhere
    \end{align*}
\end{proof}

With the correctness in place, the following lemma completes the proof of Theorem~\ref{thm:offline-approx} by bounding the running time of the offline algorithm. 

\begin{lemma}
\label{lem:offline-runtime}
    The offline algorithm runs in time $O(m\log (nW)  (\log^3 n) (\log \log n)/\epsilon)$.
\end{lemma}
\begin{proof}
    Let $m_y$ be the number of alive edges in a subproblem $y$.
    For a subproblem $[\ell, r]$, the time to find the alive edges and nodes in $G_x$ is $O(m_r)$, where $x$ is the midpoint of $[\ell, r]$. 
    For an alive edge $e = (u,v)$ in $G_x$, if $u$ is alive in $G_x$, then inserting $e$ in $G'_x$ takes $O(1)$ time.
    Otherwise, the algorithm needs to recover $\hat{d}^x(u)$ from the lowest ancestor of $x$ in the recursion tree in which $u$ is alive in order to calculate the weight of the edge $(s, v)$ in $G'_x$. 
    Each node $y$ in the recursion tree maintains $\hat{d}^{y}(w)$ for all alive vertices $w$ in $G_{y}$ in a balanced binary search tree.
    Therefore, the algorithm can check whether a vertex is alive in a subproblem in time $O(\log n)$.
    Also, if $w$ is alive in $G_y$, it takes $O(\log n)$ time to find $\hat{d}^y(w)$.
    Since node $x$ has at most $\log m$ ancestors, recovering $\hat{d}^x(u)$ can be done in $O(\log n \log\log m)$ time by doing a binary search on the ancestors.
    Hence, building $G'_x$ takes $O(m_x\log n \log\log m)$ time. 
    The time to run Dijkstra's on $G_x'$ is $O(m_x \log n)$.
    Also, the time to build the balanced binary search tree corresponding to subproblem $x$ is $O(m_x \log n)$, as the number of alive nodes in $G_x$ is bounded by the number of alive edges in $G_x$.
    Thus, the runtime of the algorithm is bounded by $O(\log n \log \log m)$ times the total number of alive edges in all the subproblems.
    
    For a given $i$ and $v$, let $x$ be the time when $\hat{d}(v)$ decreases from $( 1 + \epsilon/\log m)^{i+1}$ to $(1 + \epsilon/\log m)^i$.  
    Then $v$ is alive in any subproblem that contains $x$; there is at most 1 such subproblem in each level of the recursion tree.
    For each vertex $v$, summing over all $\log_{1 + \epsilon/\log m} (nW) = O(\log (nW) \log m/\epsilon)$ values of $i$ and the $\log m$ levels of the recursion tree, we have that there are in total $O(\log (nW) \log^2 m/\epsilon)$ subproblems in which $v$ is alive.
    Since an edge is only alive if its head is alive, there are only a total of $m \log (nW) \log^2 m/\epsilon$ alive edges over all subproblems.  
    Substituting $\log m = O(\log n)$, we obtain an aggregate running time $O(m\log (nW)(\log^3 n)(\log \log n)/\epsilon)$. 
\end{proof}
\section{Incremental Shortest Paths with Predictions}

This section gives the algorithm and analysis for the online version of the problem with predictions. Recall that in this model the edges arrive online.  The goal of the algorithm is to maintain an array $D$ of length $n$.  At any time $t$, $D[i]$ should contain a $(1 + \epsilon)$-approximation of $d^t(v_i)$.

\subsection{The Algorithm}
This section describes the online algorithm. Before any edges arrive, the algorithm is given a prediction $\hat{\sigma}$ of the edge arrival sequence.  

\paragraph{Updated Predictions.}
Our algorithm dynamically maintains the predicted sequence of edge inserts by changing $\hat{\sigma}$ online.  Specifically, after the $t$th edge arrival, the algorithm will construct a new prediction which we refer to as $\hat{\sigma}_t$.  Initially,  $\hat{\sigma}_0 = \hat{\sigma}$. We call the sequence  $\hat{\sigma}_t$ the \defn{updated prediction}. 
Intuitively, the algorithm modifies the updated prediction after each edge arrival based on which edge actually arrives. 
 
For each edge $e$, let $\widehat{\ind}_t(e)$ be the position of $e$ in $\pred_t$.  Let $\ind(e)$ denote the position of $e$ in $\sigma$ (i.e.\ the true time when $e$ arrives). If $e \notin \pred_t$, then we define $\widehat{\ind}_t(e) := m+1$. 

Our algorithm updates the sequence of edges to agree with all edges seen so far.  In other words, at time $t$, the algorithm maintains that for all edges $e$ that arrive by $t$ in $\sigma$, $\widehat{\ind}_{t}(e) = \ind(e)$.

\paragraph{Maintaining Metadata from the Offline Algorithm.}
Since the online algorithm continuously updates the result of the offline approach, it stores information to help it navigate the result of the offline algorithm.

Recall that the offline algorithm stores, for each time $t$, the set of vertices and edges alive at time $t$, as well as a distance estimate $\hat{d}^{t}(v)$ for any vertex $v$ alive at time $t$.  
The online algorithm maintains exactly the same information.  
We will see that the algorithm updates these estimates continuously as the updated prediction changes.

\paragraph{Algorithm Description.}

First, let us describe the preprocessing performed by the algorithm before any edges arrive.
The algorithm sets $D[i] = \infty$ for all $i$, except for $D[1] = 0$, which represents the source. 
Then, the algorithm runs the offline algorithm from Section~\ref{sec:offline} on $\hat{\sigma}$.   
By running the offline algorithm, the algorithm will store $\hat{d}^t(v)$ for all times $t$ and all nodes $v$ alive at time $t$.   

Now, let us describe how the algorithm runs after the $t$th edge is inserted.
At each time $t$, the algorithm \emph{rebuilds} a subset of all the subproblems.  These rebuilds update the precomputed $\hat{d}^t(v)$. When a subproblem is rebuilt, all of its descendants are rebuilt as well.  The rebuilding procedure is formally described below. 

Let $e_t$ be the edge that arrives at time $t$, and let $t'= \widehat{\ind}_{t-1}(e_t)$, i.e., $t'$ is the predicted arrival time for $e_t$ in the updated predictions immediately before it arrives. 
Note that since the first $t-1$ edges of $\sigma$ and $\pred_{t-1}$ are the same, we always have $t' \geq t$ (assuming the edges in the input sequence are distinct).

First, we describe how the algorithm updates $\hat{\sigma}_{t-1}$ to get $\hat{\sigma}_t$.
\begin{itemize}
    \item If $t = t'$, i.e., the position of $e_t$ is predicted correctly at the time it is seen, then set $\pred_t := \pred_{t-1}$.
    \item If $t \neq t'$ and $t' \leq m$; that is, $e_t$ is predicted to arrive at a later time. In this case, the algorithm moves $e_t$ from position $t'$ to position $t$ in $\hat{\sigma}_{t-1}$, and shifts everything between $t$ and $t'$ one slot to the right to obtain $\pred_t$. 
    \item If $t \neq t'$ and $t' = m + 1$; that is, at time $t-1$, $e_t$ is not in the predicted sequence.  In this case, the algorithm inserts $e_t$ in position $t$, shifts the rest of the sequence one slot to the right, and truncates the predicted sequence to length $m$ to obtain $\pred_t$. 
\end{itemize}

\paragraph{Rebuilding Subproblems.} 
Next, the algorithm \defn{rebuilds} subproblems.  Recall that the algorithm given in Section~\ref{sec:offline} is recursive; each of its recursive calls can be represented by a node in the recursion tree. 

To rebuild a subproblem $[\ell, r]$, the offline algorithm is called on $[\ell, r]$ using the updated prediction $\hat{\sigma}_t$.  The rebuild makes recursive calls as normal.
Any time a subproblem with midpoint $x$ is rebuilt, the value of $\hat{d}^{x}(v)$ is updated based on the rebuild for all alive vertices $v$.  

Let $[\ell_m, r_m]$ be the largest subproblem with $t \leq (\ell_m+r_m)/2 < t'$; then the algorithm rebuilds $[\ell_m, r_m]$. As mentioned above, all descendants of $[\ell_m, r_m]$ will be recursively called, and therefore rebuilt as well.  
In other words, the algorithm rebuilds all the subproblems $[\ell, r]$, with $t \leq (\ell+r)/2 < t'$, and all of their descendants from top to bottom (so in the first case, no subproblem gets rebuilt). 
See Figure~\ref{fig:tree_jump} for an illustration. 
Let $\rebuild(t)$ be the set of all times $t''$ such that $t''\in [\ell, r]$ for some $[\ell, r]$ rebuilt at time $t$. If no subproblem is rebuilt at time $t$, we define $\rebuild(t) := \{t\}$ to insure that $t$ is always in $\rebuild(t)$. 

\paragraph{Updating the Distance Array.}
Finally, the algorithm must update the array $D$ containing the estimated distance to each vertex. 
When a new edge is inserted, some of the entries in $D$ need to be overwritten, as their estimated distance might have changed. 
At each time $t$, we want to have $D[i] = \hat{d}^t(v_i)$ for all $i$. 
To do so, the algorithm does the following for each time $t' \in \rebuild(t)$, where $t' \leq t$, in sorted order. The algorithm iterates through all alive vertices $v_i$ in $G_{t'}$, and sets $D[i] = \hat{d}^{t'}(v_i)$. 

\begin{figure}
    \centering
    \begin{tikzpicture}[main/.style = {fill = black, circle, inner sep = 2pt}, 
    label/.style = {fill = none, circle, inner sep = 3pt},
    scale = 0.9]
        \draw[step = 1cm, gray, thin] (0,1) grid (16,5);
        \filldraw[fill = white, draw = black] (0,5) rectangle (16,4);
        \foreach \i in {0, 8}
            \filldraw[fill = white, draw = black] (\i,4) rectangle (\i+8,3); 
        \foreach \i in {0, 4, ..., 12}
            \filldraw[fill = white, draw = black] (\i,3) rectangle (\i+4,2);    
        \foreach \i in {0, 2, ..., 14}
            \filldraw[fill = white, draw = black] (\i,2) rectangle (\i+2,1);

        \foreach \i in {0, 1, 2, ..., 16}
            \node[label] at (\i, 0.5) {$G_{\i}$};    
   
        \draw[->, thick, black] (8,0.2) to [out = -150, in = -30] (4,0.2);
        \filldraw[fill = gray!60, draw = black] (0,4) rectangle (8,3);
        \filldraw[fill = gray!60, draw = black] (4,3) rectangle (8,2);
        \filldraw[fill = gray!60, draw = black] (4,2) rectangle (6,1);
        \filldraw[fill = gray!60, draw = black] (6,2) rectangle (8,1);

        \filldraw[fill = gray!20, draw = black] (0,3) rectangle (4,2);
        \filldraw[fill = gray!20, draw = black] (2,2) rectangle (4,1);
        \filldraw[fill = gray!20, draw = black] (0,2) rectangle (2,1);
        
        \node[main] at (8, 4.5) {};
        \foreach \i in {4, 12}
            \node[main] at (\i, 3.5) {};
        \foreach \i in {2, 6, 10, 14}
            \node[main] at (\i, 2.5) {};
        \foreach \i in {1, 3, ..., 16}
            \node[main] at (\i, 1.5) {};
    
        \foreach \i in {4, 12}
           \draw (8, 4.5) -- (\i, 3.5);
        \foreach \i in {2, 6}
           \draw (4, 3.5) -- (\i, 2.5);
        \foreach \i in {10, 14}
           \draw (12, 3.5) -- (\i, 2.5);
        \foreach \i in {1, 3}
           \draw (2, 2.5) -- (\i, 1.5);
        \foreach \i in {5, 7}
           \draw (6, 2.5) -- (\i, 1.5);
        \foreach \i in {9, 11}
           \draw (10, 2.5) -- (\i, 1.5);
        \foreach \i in {13, 15}
           \draw (14, 2.5) -- (\i, 1.5);  
    \end{tikzpicture}    
    \vspace{-.12in}
    \caption{An illustration of the subproblems that get rebuilt during one edge insertion. In this example, at time $t=4$, the edge $e_4$ was predicted to arrive but edge $e_8$ has arrived. So the algorithm moves edge $e_8$ from position $t'=8$ to position $t=4$. The algorithm then rebuilds all the subproblems with $t \leq (\ell+r)/2 < t'$ (colored dark gray) and their descendants (colored light gray) from top to bottom.}
    \label{fig:tree_jump}
    \vspace{-.03in}
\end{figure}
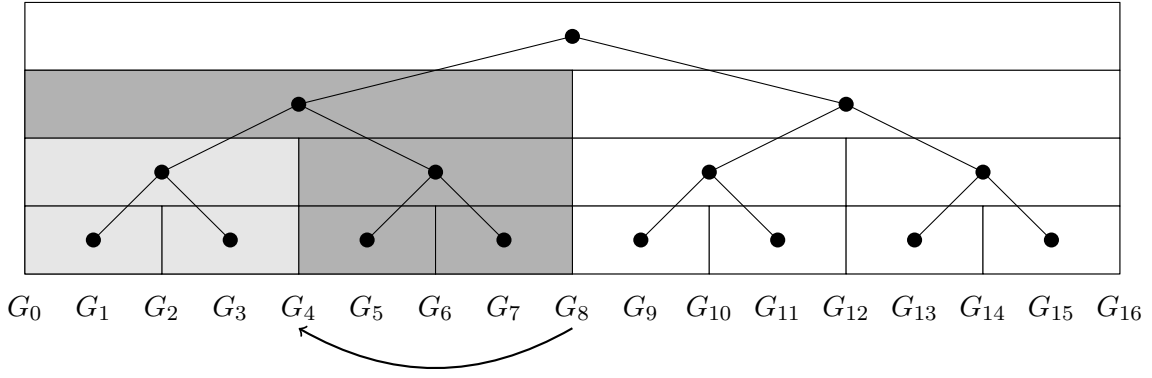

\subsection{Analysis}

This section establishes the theoretical guarantees of the algorithm. That is, the correctness of the approximation of the distances as well as the runtime bounds. 

We begin with some structure that applies to any run of the \emph{offline} algorithm.  This will help us argue correctness of our algorithm.
\begin{lemma}
\label{lem:dead_estimates_dont_change}
    For any sequence of edge inserts $\sigma'$, let $\hat{d}^t(v)$ be the distance estimates that result from running the offline algorithm on $\sigma'$.  Then 
    for any vertex $v$ and time $t \geq 1$, 
    if $v$ is dead at time $t$, then $\hat{d}^t(v) = \hat{d}^{t-1}(v)$.
\end{lemma}
\begin{proof}
    If $t = m$, then $v$ cannot be dead at time $t$.
    Otherwise, since $1 \leq t \leq m-1$, $t$ is the midpoint of some subproblem $[\ell, r]$ with $r-\ell \geq 2$, which means $\ell < t$.  If $v$ is dead at time $t$, then $\hat{d}^\ell(v) = \hat{d}^t(v)$.  
    If $\ell = t-1$, the claim is proven.
    Otherwise, we must have $t-1 \in (\ell, t)$.  Therefore, $t-1$ is the midpoint of some descendant of $[\ell, t]$ in the recursion tree.  Since $v$ is dead during all such recursive calls, $\hat{d}^{t-1}(v) = \hat{d}^\ell(v) = \hat{d}^t(v)$.
  \end{proof}
   
If the edge that arrives at time $t$ is predicted to arrive at time $t'$ according to $\pred_{t-1}$, we say the edge \emph{jumps} from $t'$ to $t$. In such case, we say it \emph{jumps over} all the positions $t \leq i < t'$. 

The next two lemmas prove the correctness of the algorithm. 

\begin{lemma}\label{lem:pred_approx}
    For any time $t$ and any vertex $v$, we have $d^t(v) \leq \hat{d}^t(v) \leq d^t(v)(1+\epsilon)$.
\end{lemma}

\begin{proof}
    For each time $t=0,\ldots,m$, let $\mathcal{T}_t$ be the recursion tree of all subproblems when the offline algorithm is run on $\hat{\sigma}_t$, and let $\hat{\mathcal{T}}_t$ be the recursion tree of the online algorithm after $t$ edge insertions. 
    In each of these trees, for each tree node $x$, the additional information of $\hat{d}^x(v)$ for all alive vertices $v$ in $G_x$ is stored in a balanced binary search tree. 
    
    We show $\hat{\mathcal{T}}_t = \mathcal{T}_t$ for all $t=0,\ldots,m$, which proves that the distance estimates for all the nodes are similar in $\hat{\mathcal{T}}_t$ and $\mathcal{T}_t$. 
    The result then follows from Lemma~\ref{lem:refined_approx} and the fact that $\pred_t(i) = \sigma(i)$ for $i = 1, \cdots, t$.   

    We prove this by induction on $t$. 
    For $t=0$, since the algorithm starts by running the offline algorithm on $\hat{\sigma}_0$, it follows that $\hat{\mathcal{T}}_0 = \mathcal{T}_0$. 
    Let $t \geq 1$ and assume $\hat{\mathcal{T}}_{t-1} = \mathcal{T}_{t-1}$.
    We want to show $\hat{\mathcal{T}}_t = \mathcal{T}_t$.
    
    At time $t$, if the newly arrived edge $e_t$ was correctly predicted, we have $\hat{\sigma}_{t}=\hat{\sigma}_{t-1}$ and the online algorithm does not change the recursion tree, meaning that $\hat{\mathcal{T}}_t=\hat{\mathcal{T}}_{t-1}$.
    Also, since $\hat{\sigma}_{t}=\hat{\sigma}_{t-1}$, we have $\mathcal{T}_t=\mathcal{T}_{t-1}$, and by the induction hypothesis, we conclude that $\hat{\mathcal{T}}_t=\mathcal{T}_t$.
    So, assume that $t'=\widehat{\ind}_{t-1}(e_t)>t$. So $e_t$ jumps from $t'$ to $t$, and the online algorithm rebuilds all the subproblems $[\ell,r]$ where $e_t$ jumps over their midpoint $x=(\ell+r)/2$ and all their descendants.  

    We show that the algorithm rebuilds all the necessary subproblems.
    In other words, we show each subproblem $x$ that is not rebuilt, i.e., is the same in $\hat{\mathcal{T}}_{t}$ and $\hat{\mathcal{T}}_{t-1}$, is also the same in $\mathcal{T}_{t}$ and $\mathcal{T}_{t-1}$.
    Since, by the induction hypothesis, $x$ is the same in trees $\hat{\mathcal{T}}_{t-1}$ and $\mathcal{T}_{t-1}$, this proves that $x$ is the same in $\hat{\mathcal{T}}_{t}$ and $\mathcal{T}_{t}$, as desired.
    
    Assume for sake of contradiction that there is an interval $[\ell, r]$  with midpoint $x$ such that $e_t$ does not jumps over the midpoint of $[\ell, r]$ or the midpoint of any of the ancestors of $[\ell, r]$, but $x$ is different in $\calT_{t-1}$ and $\mathcal{T}_t$, meaning that the binary search trees that store $\hat{d}^x(v)$ for alive vertices $v$ in $G_x$ are different in $\calT_{t-1}$ and $\mathcal{T}_t$.
    Without loss of generality, assume no ancestor of $x$ has these properties. 
    
    Since $x$ is not rebuilt, none of its ancestors is rebuilt either. By the choice of $x$, this means that all the ancestors of $x$ are the same in $\mathcal{T}_{t-1}$ and $\mathcal{T}_t$.
    The binary search tree stored at $x$ depends on the auxiliary graph $G'_x$, which in turn only depends on the alive nodes and edges in $G_x$, and the distance estimates $\hat{d}^x(u)$ for each dead vertex $u$ that is the tail of an alive edge. 
    For any vertex $u$ that is dead in $G_x$, $\hat{d}^x(u)$ is stored in the binary search tree of some ancestor $y$ of $x$, where $u$ is alive in $G_y$. 
    Since $y$ is the same in $\calT_{t-1}$ and $\calT_t$, these distance estimates are the same in $\calT_{t-1}$ and $\calT_t$.
    A node is alive in $G_x$ if and only if it is alive in $G_\ell$ and $G_r$, and its distance estimates differ at $\ell$ and $r$. 
    An edge is alive in $G_x$ in $\calT_t$ (respectively, $\calT_{t-1}$) if it is among the first $x$ edges in $\hat{\sigma}_t$ (respectively, $\hat{\sigma}_{t-1}$) , and its head is alive. 
    Since the edge $e_t$ has not jumped over $x$, we conclude that the set of the first $x$ edges in $\hat{\sigma}_t$ is similar to that of $\hat{\sigma}_{t-1}$.
    The set of alive vertices in $G_x$ only depends on the alive vertices in $G_\ell$ and $G_r$. Since $\ell$ and $r$ are ancestors of $x$, the binary search trees that store distance estimates to alive vertices in $G_\ell$ and $G_r$ are exactly the same in $\calT_{t-1}$ and $\calT_t$.
    Therefore, the set of alive vertices and edges in $G_x$ are the same in $\calT_{t-1}$ and $\calT_t$. 
    Putting everything together, we conclude that node $x$ is similar in $\calT_{t-1}$ and $\calT_t$, which contradicts the choice of $x$.  
\end{proof}

\begin{lemma}
\label{lem:D}
    After inserting edge $e_t$, we have $D[i] = \hat{d}^t(v_i)$ for $i=1,\ldots,n$.
\end{lemma}

\begin{proof}
    We show the lemma by induction on $t$. The case $t = 0$ is trivial. 
        
    If $v_i$ is dead in $G_{t'}$ for every $t'\in \rebuild(t)$, where $t' \leq t$, then $v_i$ is dead at time $t$ (note that by definition we always have $t \in \rebuild(t)$).
    Also, in this case, $D[i]$ is not overwritten at time $t$. By Lemma~\ref{lem:dead_estimates_dont_change} and the induction hypothesis, $D[i] = \hat{d}^{t-1}(v_i) = \hat{d}^t(v_i)$. 

    Otherwise, let $t'$ be the largest time less than or equal to $t$ in $\rebuild(t)$ at which $v_i$ is alive.  By definition of the algorithm, $D[i] = \hat{d}^{t'}(v_i)$.
    If $t' = t$, then $D[i] = \hat{d}^{t}(v_i)$, and if $t' < t$, then by repeatedly applying Lemma~\ref{lem:dead_estimates_dont_change}, $\hat{d}^t(v_i) = \hat{d}^{t'}(v_i)$.
\end{proof}

Now we determine the aggregate runtime of the online algorithm. Consider dividing $\sigma$ into two sets of high- and low-error edges based on an integer parameter $\tau \geq 0$.  Let $\text{LOW}(\tau) = \{e \in \sigma : |\ind(e) - \widehat{\ind}_0(e)| \leq \tau\}$ and $\text{HIGH}(\tau) =  \{e_1, \ldots, e_m\} \setminus \text{LOW}(\tau)$. 
Therefore, $\text{HIGH}(\tau)$ is the set of edges whose initial predicted arrival time is more than $\tau$ slots away from their actual arrival time.

\begin{lemma}
\label{lem:jump}
    For any integer $\tau\geq 0$ and  position $i \in \{1,\ldots,m\}$, there are at most $\tau + 2|\text{HIGH}(\tau)|$ jumps over $i$.
\end{lemma}

\begin{proof}
    Fix $\tau$. 
    Let $h_t$ be the number of edges in $\text{HIGH}(\tau)$ that have arrived by time $t$. 
    To begin, we claim that at any time $t$, each edge $e \in \text{LOW}(\tau)$ is at most $\tau + h_t$ slots ahead of its true position, i.e., $\widehat{\ind}_{t}(e) - \ind(e) \leq \tau + h_t$.
    This is true at $t = 0$ by definition. 
    For sake of contradiction, consider the first time $t$ at which this condition does not hold for some edge $e$, i.e., $\widehat{\ind}_{t-1}(e) - \ind(e) \leq \tau + h_{t-1}$, but $\widehat{\ind}_{t}(e) - \ind(e) > \tau + h_t$. 
    Notice that at each time the position of an edge can increase by at most 1, so the above can only happen if $\widehat{\ind}_{t}(e) = \widehat{\ind}_{t-1}(e) + 1$ and $h_t = h_{t-1}$, meaning the edge $e'$ inserted at time $t$ is in $\text{LOW}(\tau)$. Since insertion of $e'$ changes the position of $e$, $e'$ must have jumped from a position greater than $\widehat{\ind}_{t-1}(e)$ to position $t$. 
    Thus $\widehat{\ind}_{t-1}(e')>\widehat{\ind}_{t-1}(e)$. 
    Also note that since $e'$ has jumped over $e$ in the updated predictions, it means that $e$ has not arrived yet, i.e., $\ind(e)>t=\ind(e')$.
    Therefore, 
    \[
    \widehat{\ind}_{t-1}(e') - \ind(e') \geq \widehat{\ind}_{t-1}(e) - \ind(e) +2 = \widehat{\ind}_{t}(e) - \ind(e)+1 > \tau + h_t = \tau + h_{t-1},\]
    which contradicts the choice of $t$. 
    
    Now, fix a position $i\in \{1,\ldots,m\}$. 
    At each time $t$, either no jump happens, or some edge jumps from position $t' > t$ to position $t$.
    Consider the first time $t^*$ an edge $e \in \text{LOW}(\tau)$ jumps over $i$.
    We want to show that $t^*$ cannot be ``much smaller" than $i$.
    Assume this jump is from position $t' > i$ to position $t^* \leq i$.
    So, the position of $e$ at time $t^*-1$ in the updated prediction is $t'$, i.e., $\widehat{\ind}_{t^*-1}(e)=t'$.
    Also, the actual position of $e$ is $t^*$, i.e., $\ind(e)=t^*$.
    We know that 
    \[t'-t^*=\widehat{\ind}_{t^*-1}(e)-\ind(e) \leq \tau + h_{t^*-1}\leq \tau+|\text{HIGH}(\tau)|,\]
    which means that $i-t^*<t'-t^*\leq \tau+|\text{HIGH}(\tau)|$.
    
    Before time $t^*$, all the edges that might have jumped over $i$ are in $\text{HIGH}(\tau)$. 
    So there are at most $|\text{HIGH}(\tau)|$ jumps over $i$ before time $t^*$.
    Also, after time $i$, no edge can jump over $i$.
    Therefore, the total number of jumps over $i$ is at most $|\text{HIGH}(\tau)| + (i-t^*+1) \leq \tau + 2|\text{HIGH}(\tau)|$.   
\end{proof}

We now prove the running time guarantees of the online algorithm.  

 \begin{lemma}\label{lem:pred_runtime}
    The online algorithm runs in time $\Tilde{O}\left(m \cdot \min\limits_{\tau} \left\{\tau + |\text{HIGH}(\tau)|\right\}\cdot \log W / \epsilon \right)$.
\end{lemma}

\begin{proof}
    Each subproblem $[\ell,r]$ is rebuilt only if an edge jumps over its midpoint or the midpoint of one of its ancestors.
    Each subproblem has at most $\log m$ ancestors. 
    For each $\tau$, it follows from Lemma~\ref{lem:jump} that the subproblem $[\ell,r]$ is rebuilt at most $(\log m)(\tau + 2|\text{HIGH}(\tau)|)$ times.
    From Lemma~\ref{lem:offline-runtime}, we know it takes $\Tilde{O}(m \log W/ \epsilon)$ time to rebuild all the subproblems once. 
    Thus, the time it takes to do all the rebuilds in the online algorithm is
    $\Tilde{O}\left(m \cdot \min\limits_{\tau} \left\{\tau + |\text{HIGH}(\tau)|\right\}\cdot \log W / \epsilon \right)$.
    
    We must also account for the time required to maintain the distance array $D$.  
    Consider the updates on $D$ at time $t$. 
    If no subproblem gets rebuilt at time $t$, we only iterate through the alive vertices in $G_t$ when updating $D$. We can charge this cost to the cost of the last rebuild of subproblem $t$, as the last time subproblem $t$ was rebuilt, it had the same set of alive vertices as it has at time $t$. 
    Otherwise, in order to update $D$, we only iterate once through the alive vertices of a subset of the subproblems that get rebuilt at time $t$, and we can charge this cost to the cost of rebuild of these subproblems at time $t$. Note that this way, each subproblem that gets rebuilt throughout all the edge insertions gets charged at most once, which means that maintaining the distance array $D$ does not have any asymptotic overhead.
    
    Finally, we need to add the time needed to update the predicted sequence $\hat{\sigma}$. The total number of slots an edge $e \in \{e_1, \ldots, e_m\}$ is shifted by over all edges equals the total number of times a position $i \in \{1, \ldots, m\}$ is jumped over for all positions.  By Lemma \ref{lem:jump}, each position gets jumped over at most $\tau + 2|\text{HIGH}(\tau)|$ times for any $\tau$. Therefore, updating the predicted sequence takes $O(m \cdot (\tau + |\text{HIGH}(\tau)|))$ time for any $\tau$.
        
    Thus, the total runtime of the algorithm is $\Tilde{O}\left(m \cdot \min\limits_{\tau} \left\{\tau + |\text{HIGH}(\tau)|\right\}\cdot \log W / \epsilon \right)$.
\end{proof}

Theorem~\ref{thm:online} follows from  Lemma~\ref{lem:D} and the discussion in Section~\ref{sec:prelim} for the approximation guarantees and Lemma ~\ref{lem:pred_runtime} for the runtime.
\section{All Pairs Shortest Paths}

Our single-source  approach can be run repeatedly to approximate distances between all pairs of vertices.

Specifically, the goal is to preprocess $G_0, \ldots, G_m$ such that given $i$, $j$, and $t$, we can quickly find a $(1 + \epsilon)$-approximation of $d^t(i,j)$, the distance from $i$ to $j$ in $G_t$.  
We run the single source shortest path algorithm for each source $s\in V$, storing a separate data structure for each.  
This requires $\Tilde{O}(nm\log W/\epsilon)$ time and $\Tilde{O}(n^2 \log W / \epsilon)$ space, gives the following corollary.

\begin{corollary}
\label{cor:offline-approx-apsp}
    For the offline incremental all-pairs shortest-paths problem, there exists an algorithm running in total time $ O(nm\log(nW)\log^3 n \log\log n/\epsilon)$ that returns $(1+\epsilon)$ approximate shortest paths for each pair of vertices for each time $t$. 
\end{corollary}

\paragraph{Online Learned APSP Algorithm.}  
For the online setting, we consider the worst-case update and query bounds.  In particular, the algorithm first preprocesses $\hat{\sigma}$.  Then, it obtains the edges from $\sigma$ one by one; the time to process any such edge is the update time.  At any time during this sequence of inserts the algorithm can be queried for the distance between any two vertices in the graph; the time required to answer the query is the query time.

We can immediately combine this idea with the techniques of van den Brand et al.~\cite[Theorem 3.1]{BrandFNP24} to obtain Theorem~\ref{thm:online-apsp}.  

\begin{proof}[Proof of  Theorem~\ref{thm:online-apsp}]
We briefly summarize the algorithm of van den Brand et al.; see the proof of Theorem 3.1 in~\cite{BrandFNP24} for the full details.  
To begin, we run the algorithm from Corollary~\ref{cor:offline-approx-apsp} on $\hat{\sigma}$ in $\Tilde{O}(nm\log W / \epsilon)$ time, so that for all $t,i,j$ we can find $\hat{d}^t(i,j)$ in $O(\log \log_{1 + \epsilon} (nW))$ time.  

On a query at time $t$ for vertices $i$ and $j$, we do the following.  Let $t'$ be the latest time such that all edges that were predicted to appear by time $t'$ in $\hat{\sigma}$ have actually appeared (in $\sigma$) by time $t$.  Notice that if $\hat{\sigma}$ is a permutation of $\sigma$, then $t - t' \leq \eta$ (see the discussion immediately after Theorem 3.1 in~\cite{BrandFNP24} for a formal proof).  
Let $E'$ be the set of all edges that have arrived by time $t$ minus the edges that were predicted to arrive by time $t'$. Then, $|E'| = t - t' \leq \eta$. 
We construct $G'$ with a set of vertices equal to all vertices in any edge in $E'$, plus $i$ and $j$.  The graph $G'$ is complete.  For each edge $e = (u,v)$ in $G'$, the weight of the edge is the minimum of: (1) the weight of any edge from $u$ to $v$ in $E'$, and (2) $\hat{d}^{t'}(u,v)$.  We run Dijkstra's algorithm to obtain the distance from $i$ to $j$ in $G'$; this answers the query.  Constructing $G'$ and running Dijkstra's algorithm can be done in time $O(\eta^2 \log\log_{1 + \epsilon} (nW))$.

Let us briefly explain why this works.  
First, let's show that the answer is at least $d^t(i,j)$.  
Consider the shortest path $P$ from $i$ to $j$ in $G'$.  
Each edge $e = (u,v)$ in $P$ is one of two types: either an edge in $E'$, or an edge with weight $\hat{d}^{t'}(u,v)$.  The edge with weight $\hat{d}^{t'}(u,v)$ upper bounds the length of the shortest path between $u$ and $v$ over edges that were predicted to arrive by $t'$.  All such edges are in $G_t$ by definition of $t'$.  
Thus, by concatenating the edges from $P$ in $E'$ and the subpaths corresponding to the edges in $P$ not in $E'$, we obtain a path $P'$ in $G_t$ with weight at most that of $P$.
Now, we show the answer is at most $(1 + \epsilon)d^t(i,j)$.
Consider the shortest path $Q$ from $i$ to $j$ in $G_t$.  
We can decompose $Q$ into a sequence of subpaths, each of which is either an edge in $E'$, or consists of a sequence of edges not in $E'$.  For each edge and subpath, there is an edge in $G'$ with weight at most $(1+\epsilon)$ larger by definition; concatenating these edges we obtain a path $Q'$ in $G'$ of weight at most $(1 + \epsilon)$ larger.

With queries in mind, there are two goals for updates: we must be able to construct $E'$ efficiently, and we must maintain $t'$. 
We can store the inserted edges in a balanced binary-search tree, which takes $O(\log n)$ time per insert to maintain. At each time $t$, the BST contains $\widehat{\ind}(e)$ for all edges $e$ that have arrived by time $t$. This allows us to find $t'$ and construct $|E'|$ in $O(\log n)$ and $O(|E'| \log n) = O(\eta \log n)$ time, respectively (see the discussion in~\cite[Theorem 3.1]{BrandFNP24}).
\end{proof}
\section{Conclusion}

Learned data structures have been shown to have strong empirical performance and have the potential to be widely used in systems.  There is a need to develop an algorithmic foundation on how to leverage predictions to speed up worst-case data structures.  

In this paper, we build on this recent line of work, and provide new algorithmic techniques to solve the fundamental problem of
single-source shortest paths in incremental graphs.  As our main result, we design an ideal (consistent, robust, and smooth) algorithm for this problem.  
Our algorithm is optimal (up to log factors) with perfect predictions and circumvents the high worst-case update cost of state-of-the-art solutions even under reasonably accurate predictions. 

\bibliographystyle{plain}
\bibliography{refs}

\end{document}